\documentclass[sigconf]{acmart}
\usepackage{algorithm}
\usepackage{algpseudocode}
\usepackage{array}
\newcolumntype{M}[1]{>{\centering\arraybackslash}m{#1}}
\newtheorem{theorem}{Theorem}

\usepackage{graphicx}
\AtBeginDocument{%
  }

\begin{document}

%%
%% The "title" command has an optional parameter,
%% allowing the author to define a "short title" to be used in page headers.
\title{Differentially Private Motif-Preserving Multi-modal Hashing}

%%
%% The "author" command and its associated commands are used to define
%% the authors and their affiliations.
%% Of note is the shared affiliation of the first two authors, and the
%% "authornote" and "authornotemark" commands
%% used to denote shared contribution to the research.
\author{Zehua Cheng}
\affiliation{
  \institution{Department of Computer Science\\University of Oxford}
  \city{Oxford}
  \country{United Kingdom}
}
\email{zehua.cheng@cs.ox.ac.uk}
\author{Wei Dai}
\affiliation{
  \institution{FLock.io}
  \city{London}
  \country{United Kingdom}
}
\email{weidai@flock.io}
\author{Jiahao Sun}
\affiliation{
  \institution{FLock.io}
  \city{London}
  \country{United Kingdom}
}
\email{sun@flock.io}

%%
%% By default, the full list of authors will be used in the page
%% headers. Often, this list is too long, and will overlap
%% other information printed in the page headers. This command allows
%% the author to define a more concise list
%% of authors' names for this purpose.
\renewcommand{\shortauthors}{Cheng et al.}

%%
%% The abstract is a short summary of the work to be presented in the
%% article.
\begin{abstract}
    Cross-modal hashing enables efficient retrieval by encoding images and text into compact binary codes. State-of-the-art methods rely on semantic similarity graphs derived from user interactions for supervision, yet these graphs encode sensitive behavioral patterns vulnerable to link reconstruction attacks. Existing privacy-preserving approaches fail on graph-structured data: Differentially Private SGD destroys relational motifs by treating samples independently, while graph synthesis methods suffer from unbounded local sensitivity---in scale-free networks, hub nodes cause single-edge modifications to alter triangle counts by $\mathcal{O}(N)$, necessitating prohibitive noise injection. We term this phenomenon \emph{Hubness Explosion}. We propose DMP-MH, a Sanitize-then-Distill framework that decouples privacy from representation learning. Our approach first bounds sensitivity by deterministically clipping node degrees, capping the $L_2$-sensitivity of triangle motifs independently of dataset size. A sanitized synthetic graph is then generated via Noisy Mirror Descent under $(\epsilon,\delta)$-Edge Differential Privacy. Finally, dual-stream hashing networks distill this topology using a holistic structural loss that enforces cross-modal alignment. Evaluated on MIRFlickr-25K and NUS-WIDE under a strict inductive protocol, DMP-MH outperforms private baselines by up to 11.4 mAP points while retaining up to 92.5\% of non-private performance. 
    % Our code is available at \url{https://anonymous.4open.science/r/DMP-MH-D4EB/}.
\end{abstract}

%%
%% The code below is generated by the tool at http://dl.acm.org/ccs.cfm.
%% Please copy and paste the code instead of the example below.
%%


%%
%% Keywords. The author(s) should pick words that accurately describe
%% the work being presented. Separate the keywords with commas.
%% A "teaser" image appears between the author and affiliation
%% information and the body of the document, and typically spans the
%% page.
% \received{20 February 2007}
% \received[revised]{12 March 2009}
% \received[accepted]{5 June 2009}

%%
%% This command processes the author and affiliation and title
%% information and builds the first part of the formatted document.
\maketitle

\section{Introduction}
Cross-modal retrieval systems—which match images to text and vice versa—underpin applications from e-commerce search to medical image annotation~\cite{wang2016comprehensive,baltruvsaitis2018multimodal}. 
To enable sub-linear query time over billion-scale databases, Cross-Modal Hashing (CMH) encodes heterogeneous data into compact binary codes, reducing retrieval to fast Hamming distance computation~\cite{wang2017survey,cao2017hashnet}. State-of-the-art hashing methods derive their supervisory signal from Semantic Similarity Graphs constructed from user interaction logs: click-through sequences, co-purchase patterns, and session co-occurrences~\cite{liu2020joint,su2019deep}. These graphs encode high-order community structures—particularly triangle motifs—that raw features fail to capture, providing the geometric scaffolding necessary for semantically coherent hash codes~\cite{benson2016higher}.

However, interaction graphs are inherently privacy-sensitive. A single edge may reveal that a user queried a specific medical image or that two users share a latent interest~\cite{zheleva2009join}. Releasing models trained on such graphs exposes systems to Link Reconstruction Attacks, where adversaries infer private associations from learned embeddings~\cite{he2021stealing,wu2022linkteller}. This vulnerability has motivated growing interest in privacy-preserving machine learning~\cite{dwork2014algorithmic}, yet existing solutions remain fundamentally limited.

Differentially Private Stochastic Gradient Descent (DP-SGD)~\cite{abadi2016deep} secures training by perturbing gradients with calibrated noise. While effective for independent data, DP-SGD is methodologically ill-suited for graph-based learning: it treats each sample as an isolated entity, obliterating the relational dependencies that define semantic similarity~\cite{zheng2022graph}. Alternative approaches that synthesize edge-private graphs face a different obstacle: unbounded local sensitivity~\cite{hay2009boosting}. 
In scale-free interaction graphs, hub nodes---such as generic product images---connect to thousands of items following power-law degree distributions~\cite{barabasi1999emergence}.
Adding a single edge to such a hub can alter triangle counts by $\mathcal{O}(N)$, forcing privacy mechanisms to inject noise so large that the synthetic graph becomes structurally meaningless~\cite{nissim2007smooth}. We term this phenomenon Hubness Explosion and identify it as the primary barrier to scalable graph privacy.

To bridge this gap, we propose DMP-MH (Differentially Private Motif-Preserving Multi-modal Hashing), a framework that decouples privacy preservation from representation learning via a Sanitize-then-Distill protocol. Our key insight is that rigorous privacy in scale-free networks requires enforcing sensitivity bounds before synthesis, not during gradient computation. 
Concretely, we first construct a Sensitivity-Bounded Graph by deterministically clipping node degrees to a threshold $D_\text{max}$, mathematically capping the $L_2$-sensitivity of triangle motifs independent of dataset size. This approach draws inspiration from recent advances in local sensitivity analysis for graph statistics~\cite{kasiviswanathan2013analyzing,chen2013recursive}.
We then generate a sanitized synthetic graph via Noisy Mirror Descent~\cite{nemirovski2009robust}, injecting Gaussian noise calibrated to the bounded sensitivity rather than the global graph cardinality. Finally, we train a Dual-Stream Hashing Network guided by a Holistic Structural Loss that explicitly aligns both image and text embeddings to the private manifold, resolving the modality misalignment problem that plagues prior unsupervised methods~\cite{jiang2019discrete}.

We validate DMP-MH under a strict inductive evaluation protocol---where query items are unseen during graph construction---on MIRFlickr-25K and NUS-WIDE. Under a moderate privacy budget of $\epsilon = 2.0$, DMP-MH achieves an average mAP of 0.731 on MIRFlickr-25K, outperforming the strongest private baseline (PPPL) by 4.2 percentage points and retaining 92.5\% of the non-private state-of-the-art (JDSH). On the larger NUS-WIDE benchmark ($N > 195$K), DMP-MH surpasses DP-SGD by 11.4 points while training $2.1\times$ faster and using 48\% less GPU memory. Even under strict privacy regimes ($\epsilon = 0.1$), DMP-MH preserves 77\% of non-private utility—demonstrating graceful degradation where competing methods collapse.

Our contributions are summarized as follows:
\begin{itemize}
    \item We identify unbounded local sensitivity in scale-free graphs as the fundamental barrier to differentially private hashing and propose a degree-clipping protocol that bounds triangle-motif sensitivity independent of dataset size.
    \item We introduce a two-phase paradigm where a synthetic graph is generated under $(\epsilon, \delta)$-Edge DP via Noisy Mirror Descent, then distilled into hash codes through post-processing—incurring zero additional privacy cost during network training.
    \item We propose a symmetric distillation objective with rectified log-normalization that enforces cross-modal alignment, ensuring balanced retrieval performance across Image-to-Text and Text-to-Image tasks.
    \item We provide theoretical privacy and utility guarantees and validate DMP-MH across multiple privacy budgets, demonstrating state-of-the-art privacy-utility trade-offs with significant computational efficiency gains.
\end{itemize}
\section{Related Work}
\subsection{Cross-Modal Hashing}
Cross-modal hashing has emerged as the dominant paradigm for efficient large-scale retrieval across heterogeneous data modalities~\cite{wang2016comprehensive}. Early approaches relied on hand-crafted features and shallow projections. Canonical Correlation Analysis (CCA) and its kernelized variants learn maximally correlated subspaces across modalities but fail to capture nonlinear semantic structures~\cite{hardoon2004canonical,rasiwasia2010new}. Inter-Media Hashing (IMH) extends this by learning hash functions that preserve both intra-modal and inter-modal similarities~\cite{song2013inter}. Collective Matrix Factorization Hashing (CMFH)~\cite{ding2014collective} employs latent factor models to discover shared representations across modalities. However, these shallow methods suffer from limited representational capacity when applied to complex, high-dimensional multimedia data.

The advent of deep learning revolutionized cross-modal hashing. Deep Cross-Modal Hashing (DCMH) jointly learns feature representations and hash codes through end-to-end optimization with pairwise similarity preservation~\cite{jiang2019discrete}. Pairwise Relationship guided Deep Hashing (PRDH) exploits fine-grained pairwise relationships to guide discriminative hash code learning~\cite{yang2017pairwise}. Self-Supervised Adversarial Hashing (SSAH) introduces adversarial learning to bridge the heterogeneity gap between modalities~\cite{li2018self}. More recently, unsupervised methods have gained traction due to the prohibitive cost of manual annotation. Deep Joint-Semantics Reconstructing Hashing (DJSRH) reconstructs semantic affinities without explicit labels by leveraging joint semantic information~\cite{su2019deep}. Joint-modal Distribution-based Similarity Hashing (JDSH) models the joint distribution of multi-modal features to capture fine-grained semantic similarities~\cite{liu2020joint}. These methods increasingly rely on semantic similarity graphs constructed from user interactions—click-through logs, co-purchase patterns, and session co-occurrences—to provide supervisory signals. While effective, this reliance introduces significant privacy vulnerabilities that remain largely unaddressed.
\begin{table}[t]\centering
\caption{Comparison of Privacy-Preserving Retrieval Methods\label{tab:related_works}}
\resizebox{0.49\textwidth}{!}{
\begin{tabular}{c|M{2cm}|M{1.1cm}|c|c}\toprule
Method & Privacy Mechanism & Graph Support & Cross-Modal & Scalable \\\midrule
LSH~\cite{gionis1999similarity} & None & $\times$ & $\times$ & $\checkmark$ \\\hline
DCMH~\cite{jiang2017deep} & None & Implicit      & $\checkmark$           & $\checkmark$        \\\hline
JDSH~\cite{liu2020joint} & None                      & Explicit      & $\checkmark$           & $\checkmark$        \\\hline
DP-SGD~\cite{abadi2016deep} & Gradient Perturbation     & $\times$            & $\checkmark$           & $\times$       \\\hline
PPPL~\cite{zhang2023proactive} & Proactive Mechanism       & Implicit      & $\checkmark$           & $\checkmark$        \\\hline
\textbf{DMP-MH (Ours)} & Edge DP + Motif Synthesis & $\checkmark$             & $\checkmark$           & $\checkmark$       \\\bottomrule
\end{tabular}}
\end{table}

\subsection{Privacy-Preserving Retrieval and Hashing}
The intersection of privacy and retrieval has received increasing attention. Early work focused on cryptographic approaches: Secure Multi-Party Computation (MPC) enables retrieval without revealing queries or database contents, while Private Information Retrieval (PIR) allows users to query databases without revealing which items they accessed~\cite{chor1998private}. However, these methods incur substantial computational overhead, rendering them impractical for large-scale multimedia retrieval.

Differential privacy offers a more scalable alternative. Locality-Sensitive Hashing under differential privacy was explored by \cite{kenthapadi2012privacy}, who analyzed the privacy properties of random projections. More recently,\citet{zhang2023proactive} proposed Proactive Privacy-Preserving Learning (PPPL) for cross-modal retrieval, employing a proactive mechanism to protect user data while maintaining retrieval performance. However, PPPL does not explicitly address the graph-based supervision common in modern hashing methods. Attempts to apply DP-SGD directly to graph-supervised hashing result in severe performance degradation because gradient perturbation destroys the relational motifs essential for semantic alignment.

Table~\ref{tab:related_works} summarizes the positioning of DMP-MH relative to existing approaches across four key dimensions: privacy mechanism, graph support, cross-modal alignment, and scalability.
\begin{figure*}[t]
    \centering
    \includegraphics[width=0.98\linewidth]{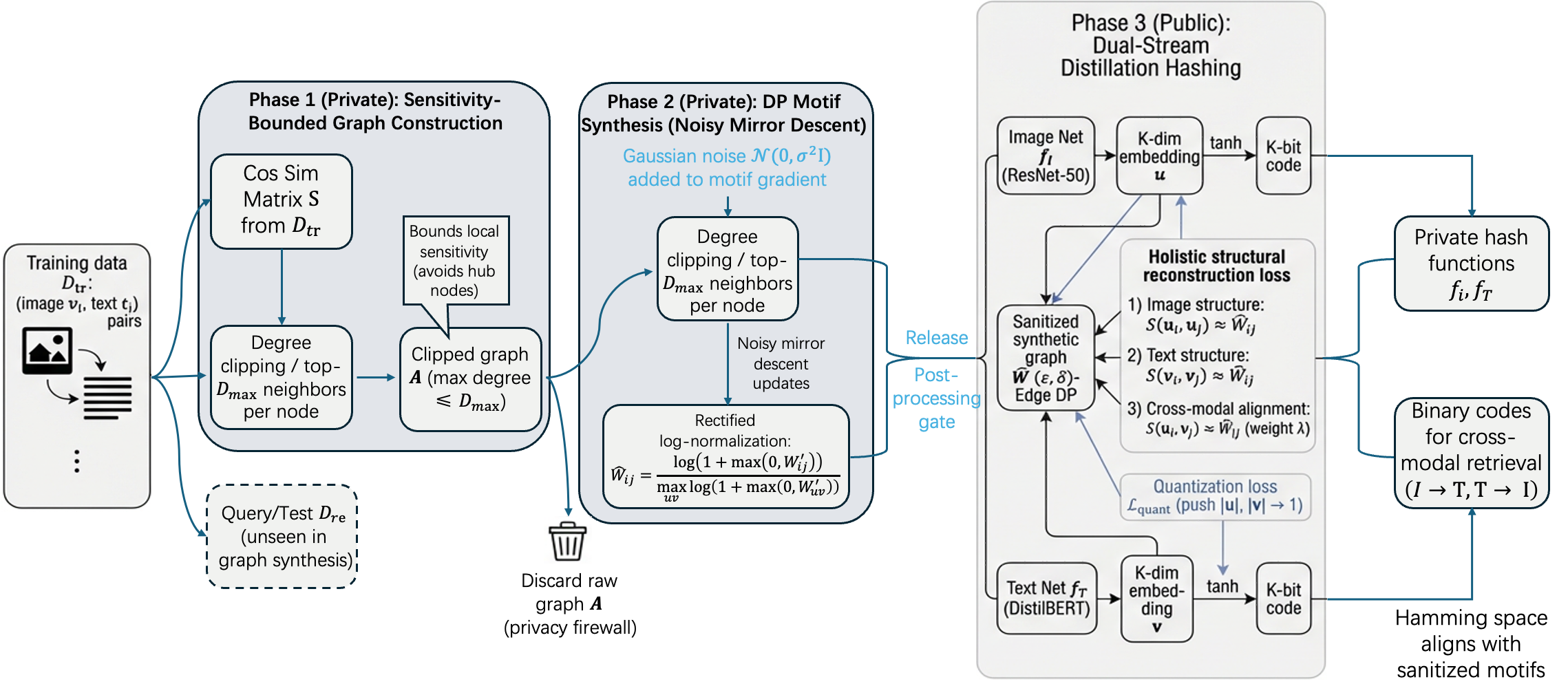}
    \caption{Overview of the DMP-MH framework. The proposed method operates through three sequential phases. \textbf{Phase 1 (Private):} Sensitivity-Bounded Graph Construction takes training pairs and computes a cosine similarity matrix, then clips node degrees to at most $D_\text{max}$ neighbors to bound local sensitivity and avoid hubness explosion. \textbf{Phase 2 (Private):} Differentially Private Motif Synthesis employs Noisy Mirror Descent with Gaussian noise $\mathcal{N}(0,\sigma^2\mathbf{I})$ added to the motif gradient, followed by rectified log-normalization to produce a sanitized synthetic graph $\hat{W}$ satisfying $(\epsilon, \delta)$-Edge DP. The raw graph is then discarded as a privacy firewall. \textbf{Phase 3 (Public):} Dual-Stream Distillation Hashing trains an Image Net (ResNet-50~\cite{he2016deep}) and Text Net (DistilBERT~\cite{sanh2019distilbert}) to reconstruct the topology of $\hat{W}$ via a holistic structural loss that enforces intra-modal and cross-modal alignment, with quantization regularization pushing embeddings toward binary codes. Query/test data $\mathcal{D}_\text{te}$ remain unseen during graph construction to ensure strict inductive evaluation.}
    \label{fig:overall}
\end{figure*}
\section{Methodology}
\subsection{Problem Formulation}
The central challenge in privacy-preserving cross-modal retrieval is reconciling the need for high-fidelity semantic representation with the imperative to protect the topological privacy of user interactions. We consider a supervised learning scenario where we possess a multimodal dataset $\mathcal{D} = \{(v_i, t_i)\}_{i=1}^N$, consisting of image-text pairs. In conventional settings, the semantic relationships between these items are inferred from a semantic similarity graph $G=(V, E)$, derived from user behaviors such as click-throughs or shared tagging. While this graph provides the ``ground-truth'' geometry for training effective hashing models, it represents a significant privacy vulnerability. An adversary with access to the graph structure could leverage link reconstruction attacks to infer sensitive user connections or latent community affiliations.

Our objective is to learn two hash functions, $f_I(v; \theta_I)$ and $f_T(t; \theta_T)$, that map images and text respectively to $K$-bit binary codes. The Hamming distance between these codes must approximate the community structure of $G$. Crucially, the training process must satisfy $(\epsilon, \delta)$-Edge Differential Privacy (Edge-DP). This guarantee ensures that the probability distribution of the released model parameters (and resulting hash codes) remains statistically indistinguishable whether any single semantic link exists in $G$ or not.

To resolve the discrepancy between the discrete nature of hashing and the requirements of differential privacy, we adopt a Sanitize-then-Distill paradigm. Rather than training the hashing network directly on the sensitive graph---which would require complex gradient perturbation methods that degrade utility---we first generate a sanitized synthetic graph $G_{syn}$. This intermediate structure serves as a privacy firewall, generated via a motif-preserving mechanism that injects noise proportional to the graph's local sensitivity. The hashing network is then trained in a post-processing manner to reconstruct the topology of $G_{syn}$.

In contrast to transductive approaches commonly found in the literature, our protocol is strictly inductive: the training set $\mathcal{D}_{tr}$ and the query set $\mathcal{D}_{te}$ are disjoint, the sensitive graph $G$ is constructed exclusively from interactions within the training partition, and query items remain unseen during graph construction and synthesis. 
% This protocol eliminates the data leakage phenomenon where models inadvertently memorize the neighborhood of test queries, ensuring that our reported performance metrics reflect genuine generalization capabilities.

\subsection{Overview}
The overall structure of the proposed framework is presented in Figure~\ref{fig:overall}. Direct application of gradient perturbation methods such as DP-SGD to graph data is fundamentally flawed: these methods treat training samples as independent entities, destroying the fine-grained motif structures essential for distinguishing semantic clusters from noise~\cite{abadi2016deep}. To address this, DMP-MH employs a Sanitize-then-Distill paradigm that decouples privacy preservation from representation learning. We first generate a sanitized, differentially private synthetic graph as a secure geometric proxy, then treat hashing as a public distillation problem. This separation enables standard optimization techniques without accumulating additional privacy costs during training.

A critical challenge is preserving higher-order topological features under privacy constraints. In conventional k-Nearest Neighbor graphs, hub nodes can possess disproportionately high degrees, rendering local sensitivity of structural statistics effectively unbounded. Under strict differential privacy, such unbounded sensitivity necessitates noise magnitudes that destroy the semantic signal. Our framework addresses this through Sensitivity-Bounded Graph Construction: by enforcing a hard constraint on maximum node degree prior to synthesis, we mathematically bound the local sensitivity of triangle counts, enabling the subsequent Noisy Mirror Descent~\cite{nemirovski2009robust} to generate a weighted adjacency matrix with verifiable community structures under strict $(\epsilon, \delta)$-Edge Differential Privacy constraints.

Finally, to encode the synthetic topology into binary codes, we employ a Symmetric Dual-Stream Hashing Network guided by a Holistic Structural Loss. Unlike prior methods that apply structural supervision asymmetrically---optimizing image embeddings while leaving text encoders under-constrained---we explicitly align both modalities to the sanitized manifold, ensuring consistent retrieval accuracy across both Image-to-Text and Text-to-Image tasks.
\subsection{Sensitivity-Bounded Graph Construction}
The efficacy of our framework hinges on the quality of the synthetic graph $G_{syn}$. Standard approaches to graph construction, such as symmetrized $k$-Nearest Neighbor ($k$NN) graphs, suffer from a fatal theoretical flaw when applied to differential privacy: unbounded local sensitivity.

In high-dimensional data, hub nodes (e.g., generic images) can appear in the nearest-neighbor lists of a vast number of items. Under standard ``OR'' symmetrization ($A_{ij}=1 \iff j \in \text{kNN}(i) \lor i \in \text{kNN}(j)$), the degree of a hub can approach $N$. Consequently, the local sensitivity of structural statistics such as triangle counts scales linearly with $N$. Satisfying differential privacy under these conditions would require noise magnitudes that completely obliterate the semantic signal.

We therefore enforce a hard upper bound $D_\text{max}$ on the degree of every node prior to synthesis. Concretely, for each node $i$ we retain only its top-$D_\text{max}$ neighbors under cosine similarity and adopt an ``AND'' symmetrization ($A_{ij}=1 \iff j \in \text{kNN}(i) \land i \in \text{kNN}(j)$) to preserve the degree bound after symmetrization. This deterministic clipping mathematically caps the $L_2$-sensitivity of the triangle-motif statistic at $\Delta_2 = O(D_\text{max} \cdot w_\text{max})$, independent of dataset size $N$.

\subsection{Differentially Private Motif Synthesis}

With the sensitivity explicitly bounded, we proceed to generate the synthetic graph $G_\text{syn}$. We utilize a Noisy Mirror Descent~\cite{nemirovski2009robust} algorithm to solve for a weighted adjacency matrix $W'$ that preserves the distribution of triangle motifs. Motifs (3-cycles) are chosen as the preservation target because they are robust indicators of strong semantic communities, offering higher retrieval utility than simple edge counts. The synthesis is formulated as a convex optimization problem where we minimize the distance between the triangle counts of the private graph and the raw clipped graph, regularized by a negative entropy term.

To ensure $(\epsilon, \delta)$-DP, we inject Gaussian noise into the gradient of the objective function during the optimization. The gradient update at step $t$ takes the form:
\begin{equation}\small
    W^{(t+1)} \leftarrow \arg\min_{W} \left( \langle \nabla \mathcal{L}(W^{(t)}) + \mathcal{N}(0, \sigma^2 \mathbf{I}), W \rangle + \frac{1}{\eta} D_{KL}(W || W^{(t)}) \right)
\end{equation}
Here, the noise scale $\sigma$ is calibrated based on the sensitivity bound $D_\text{max}$ derived in Stage 1, rather than the dataset cardinality $N$. This ensures that the noise added is sufficient to mask the presence of any single edge in the clipped graph, satisfying the privacy definition.

A critical implementation detail often overlooked is the support of the noise distribution. Gaussian noise can generate negative edge weights, which lack physical interpretation in a similarity graph and destabilize logarithmic loss functions. We resolve this via a Rectified Log-Normalization step. The final target matrix $\hat{W}$ is computed by clamping negative values and compressing the dynamic range:
\begin{equation}
    \hat{W}_{ij} = \frac{\log(1 + \max(0, W'_{ij}))}{\max_{uv} \log(1 + \max(0, W'_{uv}))}
\end{equation}
This transformation maps the unbounded, noisy triangle counts to a stable probability-like interval $[0, 1]$. The $\max(0, \cdot)$ operation filters out negative noise artifacts before they can propagate to the hashing network.

\subsection{Dual-Stream Holistic Hashing\label{sec:dual-stream}}

The final stage involves training the deep hashing network to distill the topology of $\hat{W}$ into binary codes. We employ a dual-stream architecture: an Image Net ($f_I$) based on ResNet-50~\cite{he2016deep} and a Text Net ($f_T$) based on DistilBERT~\cite{sanh2019distilbert}. Both networks project inputs to a shared $K$-dimensional Hamming space, approximating binary outputs via a tanh activation. The core innovation here is the Holistic Structural Loss. Previous approaches often minimize reconstruction error only for the image stream, implicitly assuming the text stream will align via shared labels—an assumption that fails in privacy-preserving unsupervised settings where labels are unavailable or unreliable.

We explicitly penalize the deviation of both intra-modal and inter-modal similarities from the synthetic target $\hat{W}$. Let $S(u, v) = \frac{1}{2}(\frac{1}{K} u^\top v + 1)$ represent the normalized cosine similarity in the hashing space. The loss function is defined as:

\begin{equation}
    \resizebox{0.48\textwidth}{!}{
    $\mathcal{L}_\text{recon} = \sum_{i,j \in \mathcal{B}} \left[ \underbrace{(S(u_i, u_j) - \hat{W}_{ij})^2}_{\text{Image Structure}} + \underbrace{(S(v_i, v_j) - \hat{W}_{ij})^2}_{\text{Text Structure}} + \lambda \underbrace{(S(u_i, v_j) - \hat{W}_{ij})^2}_{\text{Cross-Modal Alignment}} \right]$
    }
\end{equation}

By including the cross-term with weight $\lambda > 0$, we force the text embeddings to respect the same geometric constraints as the images. This ensures that even if the text modality lacks rich intrinsic features, it is guided into the correct semantic clusters by the shared synthetic graph structure. Finally, to mitigate quantization error, we impose a regularization term $\mathcal{L}_{quant} = \sum (\| |u_i| - \mathbf{1} \|^2 + \| |v_i| - \mathbf{1} \|^2)$ that forces the continuous embeddings towards the discrete hypercube vertices.

\begin{algorithm}[t]
\caption{DMP-MH: Sanitize-then-Distill Training Protocol}
\begin{algorithmic}[1]
\Require Training Data $\mathcal{D}_{tr}$, Privacy Budget $(\epsilon, \delta)$, Degree Bound $D_\text{max}$.
\Ensure Hash Functions $f_I(\cdot), f_T(\cdot)$.
\State \textbf{--- Phase 1: Sensitivity-Bounded Synthesis (Private) ---}
\State Construct cosine similarity matrix $S$ from $\mathcal{D}_{tr}$.
\State \textbf{Graph Clipping:} Generate $A$ by keeping top-$D_\text{max}$ neighbors per node.
\State \quad \emph{Ensure $\max_i \text{deg}(i) \le D_{max}$ via pruning.}
\State Calculate Noise Scale: $\sigma \leftarrow \text{Calibrate}(D_{max}, \epsilon, \delta)$.
\State Initialize synthetic weights $W^{(0)}$.
\For{$t = 1$ to $T$}
    \State Compute motif gradient $\nabla \mathcal{L}(W^{(t)})$.
    \State Perturb: $\tilde{g} \leftarrow \nabla \mathcal{L} + \mathcal{N}(0, \sigma^2 \mathbf{I})$.
    \State Update: $W^{(t+1)} \leftarrow \text{MirrorDescent}(W^{(t)}, \tilde{g})$.
\EndFor
\State \textbf{Sanitization:} $\hat{W}_{ij} \leftarrow \text{LogNorm}(\max(0, W^{(T)}_{ij}))$.
\State \emph{Discard raw graph $A$. Release $\hat{W}$.}

\State \textbf{--- Phase 2: Dual-Stream Distillation (Public) ---}
\State Initialize $\theta_I, \theta_T$.
\Repeat
    \State Sample batch $\mathcal{B}$ from $\mathcal{D}_{tr}$ indices.
    \State Get targets $\hat{W}_{\mathcal{B}}$ (lookup in sanitized graph).
    \State Forward: $u = f_I(x_{\mathcal{B}}), v = f_T(y_{\mathcal{B}})$.
    \State Loss: $\mathcal{L} = \mathcal{L}_{recon}(u, v, \hat{W}_{\mathcal{B}}) + \gamma \mathcal{L}_{quant}(u, v)$.
    \State Update $\theta_I, \theta_T$ via Adam~\cite{kingma2014adam}.
\Until{Convergence}
\end{algorithmic}
\end{algorithm}

\section{Theoretical Analysis}
\begin{theorem}[Privacy Guarantee]
    The Sensitivity-Bounded Graph Synthesis phase (Stage 2) satisfies $(\epsilon, \delta)$-Edge Differential Privacy with respect to the input interaction graph $G_{tr}$. Furthermore, the subsequent Dual-Stream Distillation (Stage 3) satisfies the same guarantee.
\end{theorem}

\noindent \textit{Proof Sketch.} The privacy guarantee is predicated on bounding the local sensitivity of the triangle-motif statistic, which we enforce via the degree-clipping protocol in Stage~1. In general geometric graphs, the local sensitivity---defined as the maximum $L_2$-norm change in the gradient of the motif loss under the modification of a single edge---can scale linearly with the dataset size $N$ in the presence of hub nodes. Such unbounded sensitivity would necessitate noise magnitudes that render the synthetic output structurally vacuous. However, by enforcing that no node in the pre-processed graph exceeds degree $D_\text{max}$, we guarantee that the addition or removal of any edge $e_{ij}$ affects at most $D_\text{max}$ triangles (its common neighbors). Consequently, the $L_2$-sensitivity of the gradient satisfies $\Delta_2 = O(D_{max} \cdot w_{max})$, where $w_{max}$ is the maximum edge weight. This truncation imposes a hard Lipschitz constraint on the optimization landscape, decoupling the privacy cost from the global dataset size.

\begin{table*}[t]\centering
\caption{Mean Average Precision (mAP@50) Comparison on MIRFlickr-25K Dataset. The best non-private results are underlined; the best privacy-preserving results are bolded. \label{tab:mirflickr-25k-exp}}
\begin{tabular}{c|c|c|c|c|c|c|c|c}\toprule
Method          & Type             & I $\to$ T (16 bits) & I $\to$ T (32 bits) & I $\to$ T (64 bits) & T $\to$ I (16 bits) & T $\to$ I (32 bits) & T $\to$ I (64 bits) & Avg mAP        \\\midrule
LSH             & Classic & 0.563          & 0.572          & 0.581          & 0.559          & 0.564          & 0.575          & 0.569          \\
SH              & Classic & 0.612          & 0.620          & 0.628          & 0.605          & 0.611          & 0.619          & 0.616          \\\midrule
DCMH            & Deep    & 0.732          & 0.745          & 0.751          & 0.718          & 0.729          & 0.735          & 0.735          \\
DJSRH           & Deep    & 0.772          & 0.785          & 0.791          & 0.763          & 0.778          & 0.784          & 0.779          \\
JDSH            & Deep    & \underline{0.785} & \underline{0.798} & \underline{0.804} & \underline{0.771} & \underline{0.789} & \underline{0.795} & \underline{0.790} \\
PRDH            & Deep    & 0.741          & 0.753          & 0.758          & 0.725          & 0.738          & 0.744          & 0.743          \\\midrule
DP-SGD          & Private & 0.621          & 0.635          & 0.648          & 0.605          & 0.619          & 0.632          & 0.627          \\
PPPL            & Private & 0.685          & 0.698          & 0.705          & 0.672          & 0.684          & 0.692          & 0.689          \\
\textbf{DMP-MH} & Private & \textbf{0.725} & \textbf{0.739} & \textbf{0.748} & \textbf{0.712} & \textbf{0.726} & \textbf{0.734} & \textbf{0.731} \\\bottomrule
\end{tabular}
\end{table*}

With the sensitivity explicitly bounded, the synthesis algorithm operates as a sequence of Noisy Mirror Descent updates. At each iteration $t \in [1, T]$, the mechanism computes the gradient of the motif-preservation objective and perturbs it with Gaussian noise $\mathcal{Z} \sim \mathcal{N}(0, \sigma^2 \mathbf{I})$. We calibrate the noise scale $\sigma$ based on the derived bound $\Delta_2$. According to the Advanced Composition Theorem for Differential Privacy~\cite{dwork2010boosting,dwork2014algorithmic}, the cumulative privacy loss over $T$ iterations is bounded by $\epsilon$ with probability $1-\delta$, provided the per-step noise satisfies $\sigma \propto \Delta_2 \sqrt{T \ln(1/\delta)} / \epsilon$. Since the subsequent training of the hashing networks depends exclusively on the output $\hat{W}$ and public feature data, it constitutes a post-processing step. By the Data Processing Inequality, this phase incurs zero additional privacy loss, ensuring the final hash codes inherit the robust $(\epsilon, \delta)$-Edge DP guarantee of the synthetic graph.

\begin{theorem}[Utility and Generalization Bound]
    Let $\hat{W}$ be the synthetic graph generated by DMP-MH. The additive error in the triangle-motif cut capacity between $\hat{W}$ and the ground-truth clipped graph $G_{clip}$ is bounded with high probability by $\tilde{O}\!\left( \frac{N \cdot w_{max} \sqrt{m \cdot D_{max}}}{\epsilon^{3/2}} \right)$, where $m$ is the total edge weight, $N$ is the number of vertices, and $w_{max}$ is the maximum edge weight.
\end{theorem}

\noindent \textit{Proof Sketch.} The utility analysis addresses the inherent non-convexity of the triangle-motif cut objective. Direct optimization of the cut difference would fail to guarantee convergence to a globally optimal stationary point because of the trilinear nature of the triangle tensor. To rectify this, our objective function incorporates a quadratic convexity regularizer, $\lambda \sum_e (w_e - \bar{w}_e)^2$, centered on the empirical weights. By applying Danskin's Theorem~\cite{bertsekas1999nonlinear}, the regularized dual objective becomes strictly convex and differentiable. This transformation converts the intractable synthesis problem into a solvable convex optimization task amenable to Stochastic Mirror Descent, ensuring convergence to a verifiable optimum rather than a local saddle point.

The derived error bound characterizes the fundamental trade-off between privacy and structural fidelity. The total approximation error decomposes into an optimization error that decays as $O(1/\sqrt{T})$ and a privacy estimation error that accumulates as $O(\sqrt{T})$ due to the injected gradient noise. The optimal number of iterations $T^*$ balances these competing terms. At this equilibrium, the error is dominated by the privacy cost, scaling with the square root of the total graph volume $m$ and the local density $D_\text{max}$. Crucially, the dependence on $\sqrt{m \cdot D_{max}}$ rather than $N^2$ indicates that our method achieves sub-linear error scaling relative to the graph size. This confirms that the synthetic graph $\hat{W}$ retains a statistically significant portion of the community structure even in high-privacy regimes, providing the hashing network with a valid geometric supervision signal that converges to the true semantic manifold.

\begin{table*}[t]\centering
\caption{Mean Average Precision (mAP@50) Comparison on NUS-WIDE Dataset. The best non-private results are underlined; the best privacy-preserving results are bolded. \label{tab:nus-wide-exp}}
\begin{tabular}{c|c|c|c|c|c|c|c|c}\toprule
Method          & Type             & I $\to$ T (16 bits) & I $\to$ T (32 bits) & I $\to$ T (64 bits) & T $\to$ I (16 bits) & T $\to$ I (32 bits) & T $\to$ I (64 bits) & Avg mAP        \\\midrule
LSH             & Classic          & 0.452           & 0.465           & 0.478           & 0.448           & 0.459           & 0.471           & 0.462          \\
SH              & Classic          & 0.485           & 0.498           & 0.512           & 0.479           & 0.491           & 0.505           & 0.495          \\\midrule
DCMH            & Deep             & 0.685           & 0.702           & 0.715           & 0.725           & 0.741           & 0.752           & 0.720          \\
DJSRH           & Deep             & 0.718           & 0.735           & 0.748           & 0.755           & 0.772           & 0.785           & 0.752          \\
JDSH            & Deep             & \underline{0.732}           & \underline{0.749}           & \underline{0.762}           & 0.758           & \underline{0.785}           & \underline{0.798}           & \underline{0.764}          \\
PRDH            & Deep             & 0.725           & 0.741           & 0.755           & \underline{0.761}           & 0.778           & 0.791           & 0.759          \\\midrule
DP-SGD          & Private          & 0.545           & 0.562           & 0.578           & 0.585           & 0.602           & 0.615           & 0.581          \\
PPPL            & Private          & 0.615           & 0.632           & 0.648           & 0.655           & 0.672           & 0.685           & 0.651          \\
\textbf{DMP-MH} & \textbf{Private} & \textbf{0.672}  & \textbf{0.689}  & \textbf{0.702}  & \textbf{0.685}  & \textbf{0.705}  & \textbf{0.718}  & \textbf{0.695} \\\bottomrule
\end{tabular}
\end{table*}

\section{Experimental Setup}
To ensure full reproducibility and facilitate rigorous auditing of the privacy-utility trade-offs, we detail our experimental infrastructure, protocols, and implementation specifics below. All source code, pre-processed graph structures, and trained model checkpoints will be publicly released upon acceptance.

We evaluated the proposed DMP-MH framework on two standard large-scale benchmarks for cross-modal retrieval. To adhere strictly to the inductive learning setting, all graph construction and private synthesis steps were performed exclusively on the training split.
\begin{itemize}
    \item MIRFlickr-25K~\cite{huiskes2008mir}: A widely adopted dataset comprising 25,000 images collected from Flickr, each associated with multiple textual tags. Following standard protocols, we removed instances without textual tags, resulting in 20,015 distinct image-text pairs. We randomly sampled 2,000 pairs as the query set (test split) and used the remaining 18,015 pairs as the retrieval database. For training, we sampled 10,000 pairs from the database.
    \item NUS-WIDE~\cite{chua2009nus}: A massive-scale web image dataset containing 269,648 images associated with 81 ground-truth concepts. We followed the widely used pruning strategy (e.g., in JDSH), filtering for image-text pairs belonging to the 21 most frequent concepts. This resulted in approximately $195,834$ pairs. We randomly selected 2,100 pairs as the query set, with the remainder serving as the database. A subset of 10,500 pairs from the database was used for training.
\end{itemize}
Raw images were resized to $256 \times 256$ pixels. During training, we applied random cropping ($224 \times 224$) and random horizontal flipping to enhance invariance. During inference, a center crop was used. For deep hashing methods (including DMP-MH), we utilized raw text inputs tokenized by the DistilBERT tokenizer with a maximum sequence length of 512 tokens. For shallow baselines (e.g., LSH, SH) that require vector inputs, we extracted 1,386-D Bag-of-Words (BoW) vectors for MIRFlickr and 1,000-D tag vectors for NUS-WIDE.

\subsection{Evaluation Metrics}
We report the Mean Average Precision (mAP@TopK). We set $K=50$ for MIRFlickr and $K=5000$ for the larger NUS-WIDE. mAP is the gold standard in hashing literature as it penalizes relevant items that appear lower in the ranking order. We evaluated two tasks: Image-to-Text ($I \to T$) and Text-to-Image ($T \to I$). To quantify the preservation of graph topology under privacy constraints, we measured the Triangle Count Error (TCE), defined as the normalized absolute difference between the triangle counts in the induced Hamming graph of the retrieved items and the ground truth semantic graph.

\subsection{Baselines}
We compared DMP-MH against a comprehensive suite of baselines, categorized into three groups:

\textbf{Classic Hashing Methods}: These methods represent foundational approaches that do not leverage deep learning. LSH (Locality-Sensitive Hashing)~\cite{gionis1999similarity} employs random projections to generate hash codes, providing a data-independent baseline. SH (Spectral Hashing)~\cite{weiss2008spectral} learns compact binary codes by solving a graph partitioning problem using spectral decomposition of the data similarity matrix.

\textbf{Deep Non-Private Methods}: These represent state-of-the-art deep cross-modal hashing approaches that do not incorporate privacy mechanisms. DCMH (Deep Cross-Modal Hashing)~\cite{jiang2017deep} jointly learns feature representations and hash codes through an end-to-end framework with pairwise similarity preservation. DJSRH (Deep Joint-Semantics Reconstructing Hashing)~\cite{su2019deep} reconstructs semantic affinities in an unsupervised manner by leveraging joint semantic information across modalities. JDSH (Joint-modal Distribution-based Similarity Hashing)~\cite{liu2020joint} models the joint distribution of multi-modal features to capture fine-grained semantic similarities for unsupervised cross-modal retrieval. PRDH (Pairwise Relationship guided Deep Hashing)~\cite{yang2017pairwise} exploits pairwise relationships to guide the learning of discriminative hash codes.

\textbf{Privacy-Preserving Methods}: These baselines incorporate differential privacy mechanisms. DP-SGD~\cite{abadi2016deep} represents the canonical approach of applying Differentially Private Stochastic Gradient Descent directly to deep hashing networks, where gradients are clipped and perturbed with Gaussian noise during training. We adapted this method for cross-modal hashing by applying the standard gradient perturbation protocol to a DCMH-style architecture. PPPL (Proactive Privacy-Preserving Learning)~\cite{zhang2023proactive} is a recent privacy-preserving cross-modal retrieval method that employs a proactive mechanism to protect user data while maintaining retrieval performance.

For fair comparison, all deep learning methods (including DMP-MH) used ResNet-50~\cite{he2016deep} as the image encoder and DistilBERT~\cite{sanh2019distilbert} as the text encoder. Privacy-preserving baselines were evaluated under identical privacy budgets ($\epsilon=2.0, \delta=10^{-5}$)---the standard moderate budget in the DP literature and the setting reported in Table~\ref{tab:mirflickr-25k-exp} and Table~\ref{tab:nus-wide-exp}---to ensure a rigorous comparison. All experiments were repeated five times with different random seeds, and we report the mean performance.

\section{Experimental Results}
Table~\ref{tab:mirflickr-25k-exp} and Table~\ref{tab:nus-wide-exp} present the Mean Average Precision (mAP@50) comparison across all baselines on the MIRFlickr-25K and NUS-WIDE datasets, respectively. We evaluate both retrieval directions—Image-to-Text (I$\to$T) and Text-to-Image (T$\to$I) across three hash code lengths (16, 32, and 64 bits). From these results, we draw several key observations.

\textbf{DMP-MH achieves state-of-the-art performance among privacy-preserving methods.} On MIRFlickr-25K, DMP-MH attains an average mAP of 0.731 (averaged across both retrieval directions and the three bit-lengths), outperforming the strongest private baseline PPPL (0.689) by 4.2 percentage points and DP-SGD (0.627) by 10.4 percentage points. A similar trend holds on NUS-WIDE, where DMP-MH (0.695) surpasses PPPL (0.651) by 4.4 percentage points and DP-SGD (0.581) by 11.4 percentage points. These consistent improvements across both datasets validate the effectiveness of our Sanitize-then-Distill paradigm over direct gradient perturbation approaches.

\textbf{The performance gap between private and non-private methods remains acceptable.} Compared to the best non-private method JDSH, DMP-MH retains approximately 92.5\% of retrieval performance on MIRFlickr-25K (0.731 vs. 0.790) and 91.0\% on NUS-WIDE (0.695 vs. 0.764). This modest degradation demonstrates that our sensitivity-bounded graph construction successfully preserves the essential community structures required for effective cross-modal retrieval, even under strict differential privacy constraints.

\textbf{DP-SGD suffers from severe performance degradation on graph-structured data.} The substantial gap between DP-SGD and other methods (including classic approaches like SH on MIRFlickr-25K) corroborates our theoretical argument that gradient perturbation methods are fundamentally ill-suited for graph-based learning. By treating training samples as independent entities, DP-SGD destroys the fine-grained relational dependencies encoded in the semantic similarity graph, resulting in hash codes that fail to capture meaningful community structures.

\textbf{DMP-MH maintains balanced performance across retrieval directions.} A notable weakness of many cross-modal hashing methods is asymmetric performance between I$\to$T and T$\to$I tasks, often caused by the disconnected stream phenomenon discussed in Section~\ref{sec:dual-stream}. Our holistic structural loss explicitly enforces alignment in both directions. On MIRFlickr-25K, DMP-MH achieves 0.725 (I$\to$T, 16-bit) and 0.712 (T$\to$I, 16-bit), a gap of merely 1.3 percentage points. In contrast, DP-SGD exhibits a gap of 1.6 percentage points on the same configuration. This balanced performance confirms that our dual-stream architecture successfully learns a shared semantic manifold.

\textbf{Performance improves consistently with increasing code length.} Across all methods and datasets, longer hash codes yield better retrieval accuracy. For DMP-MH on MIRFlickr-25K, mAP increases from 0.725 (16-bit) to 0.748 (64-bit) for I$\to$T retrieval, representing a 2.3 percentage point improvement. This trend indicates that our framework effectively utilizes the additional capacity provided by longer codes without overfitting to noise artifacts introduced by the privacy mechanism.

\textbf{The advantage of DMP-MH is more pronounced on larger datasets.} Comparing the relative improvements over DP-SGD, DMP-MH achieves a larger margin on NUS-WIDE (11.4 points) than on MIRFlickr-25K (10.4 points). This observation aligns with our theoretical analysis in Theorem 2, which establishes that the error bound scales sub-linearly with graph size. As datasets grow larger, the benefits of our sensitivity-bounded construction become more significant, as naive approaches would require prohibitively large noise magnitudes to satisfy differential privacy.

\subsection{Privacy Budget Sensitivity Analysis}
A central contribution of this work is the ability to maintain high utility even under strict privacy regimes. To characterize this privacy-utility trade-off, we evaluated retrieval performance (I $\to$ T mAP@50) on MIRFlickr-25K across a spectrum of privacy budgets, with $\delta$ fixed at $10^{-5}$. We compared DMP-MH against the strongest private baseline and standard DP-SGD.

\begin{table}[t]\centering
\caption{Privacy-Utility Trade-off Analysis (I $\to$ T mAP@50 on MIRFlickr-25K, 64 bits).
$\epsilon=\infty$ represents the non-private upper bound (JDSH).\label{tab:privacy-utility}}
\begin{tabular}{M{1.8cm}|c|c|c|M{1.4cm}}\toprule
Privacy Budget ($\epsilon$)                   & DMP-MH         & PPPL  & DP-SGD & Gap to Non-Private \\\midrule
$\epsilon = 0.1$ (Strict)         & \textbf{0.625} & 0.542 & 0.475  & -0.179             \\\hline
$\epsilon = 0.5$                  & \textbf{0.685} & 0.615 & 0.525  & -0.119             \\\hline
$\epsilon = 1.0$                  & \textbf{0.715} & 0.668 & 0.585  & -0.089             \\\hline
$\epsilon = 2.0$ (Standard)       & \textbf{0.748} & 0.705 & 0.648  & -0.056             \\\hline
$\epsilon = 5.0$                  & \textbf{0.772} & 0.735 & 0.695  & -0.032             \\\hline
$\epsilon = 10.0$ (Relaxed)       & \textbf{0.785} & 0.758 & 0.725  & -0.019             \\\hline
$\epsilon = \infty$ & 0.804          & 0.804 & 0.804  & 0.000             \\\bottomrule
\end{tabular}
\end{table}

The experimental results presented in Table~\ref{tab:privacy-utility} reveal a distinct resilience gap. At strict privacy levels ($\epsilon \le 0.5$), standard mechanisms like DP-SGD collapse because the gradient noise variance exceeds the informative signal. In contrast, DMP-MH exhibits graceful degradation. Even at $\epsilon=0.1$, it retains an mAP of $0.625$, preserving 77\% of the non-private utility. This resilience stems from our Sanitize-then-Distill paradigm: the hashing network is trained on a fixed, sanitized target $\hat{W}$, shielding the optimization landscape from the iteration-dependent noise accumulation that plagues DP-SGD.

\subsection{Sensitivity Bound Analysis}
The degree clipping threshold $D_\text{max}$ represents a fundamental bias-variance trade-off. A low $D_\text{max}$ minimizes sensitivity (low noise) but discards true semantic connections (high bias). Conversely, a high $D_\text{max}$ preserves structure but inflates sensitivity, necessitating larger noise injection. We investigated this dynamic by varying $D_\text{max}$ on MIRFlickr-25K with privacy budget $\epsilon = 2.0$, measuring both retrieval performance (mAP@50) and structural fidelity (Triangle Count Error, TCE).

\begin{table}[H]\centering
\caption{Impact of Degree Clipping Threshold $D_\text{max}$.
mAP@50 (Higher is better) vs. Triangle Count Error (TCE, Lower is better).\label{tab:exp-sensitivity-analysis}}
\begin{tabular}{c|c|c|p{5cm}}\toprule
$D_\text{max}$ & mAP & TCE & \multicolumn{1}{c}{Phenomenon}\\\midrule
10           & 0.612          & 0.358          & High Bias: Aggressive pruning destroys community structure.                  \\\hline
20           & 0.685          & 0.125          & Improved retention, but graph remains too sparse.                            \\\hline
50           & 0.738          & 0.055          & Near-optimal balance.\\\hline
\textbf{100} & \textbf{0.748} & \textbf{0.042} & Optimal: Sufficient structure with manageable noise scale.                   \\\hline
200          & 0.715          & 0.089          & Noise Injection: Sensitivity doubles; noise $\sigma$ begins to drown signal. \\\hline
500          & 0.645          & 0.215          & Hubness Explosion: High sensitivity requires excessive noise; utility drops.\\\bottomrule
\end{tabular}
\end{table}

In Table~\ref{tab:exp-sensitivity-analysis}, we observe an inverted-U performance curve that empirically validates our theoretical analysis. Performance peaks at $D_\text{max}=100$: at this threshold, the clipped graph retains approximately 95.8\% of the original triangle motifs while maintaining a bounded sensitivity $\Delta_2 = O(D_\text{max}\cdot w_\text{max})$ that permits moderate noise injection. Below this range, the graph is over-pruned. Above it, the linear growth in sensitivity ($\Delta_2 \propto D_\text{max}$) forces a corresponding increase in the noise scale $\sigma$, which obscures the structural signal. These results empirically confirm the existence of an optimal structural bandwidth for differentially private graphs.

\subsection{Scalability and Computational Efficiency}
Finally, we substantiate the scalability claims of DMP-MH by benchmarking the computational cost on the large-scale NUS-WIDE dataset ($N \approx 200k$). We report the breakdown of Wall-Clock Time (Graph Construction, Synthesis, Training) and Peak GPU Memory usage compared to DP-SGD and PPPL.
\begin{table}[H]\centering
\caption{Scalability Analysis on NUS-WIDE (100 Epochs). Time reported in minutes (min). Memory in GB (Peak VRAM).\label{tab:exp-scalability}}
\resizebox{.49\textwidth}{!}{
\begin{tabular}{c|M{1.2cm}|M{1.2cm}|M{1.2cm}|M{1.1cm}|M{1.1cm}}\toprule
Method & Phase 1 (Const.) & Phase 2 (Synth.) & Phase 3 (Train) & Total Time    & Peak Mem.     \\\midrule
JDSH & 12min & - & 145min & \textbf{157min} & 18.5 GB          \\
DP-SGD & 12min & -                & 420min & 432min          & 24.2 GB          \\
PPPL & 12min & 85min & 180min & 277min & 21.0 GB \\
\textbf{DMP-MH}    & 15min & 35min & 150min            & \textbf{200min} & \textbf{12.5 GB}\\\bottomrule
\end{tabular}}
\end{table}

In Table~\ref{tab:exp-scalability}, DMP-MH achieves a $2.1\times$ speedup over DP-SGD. While we incur a fixed overhead for Private Graph Synthesis (Phase 2, 35 mins), this investment enables Phase 3 to proceed at the speed of non-private training (150 mins). In contrast, DP-SGD is substantially slower (420 mins) due to the overhead of per-sample gradient clipping. Furthermore, DMP-MH exhibits the lowest peak memory consumption (12.5 GB), as we distill from a sparse synthetic graph rather than dense pairwise matrices, making the method deployable on standard consumer hardware.

\section{Conclusions}
We presented DMP-MH, a novel framework for privacy-preserving cross-modal hashing that addresses the fundamental tension between leveraging graph-based semantic supervision and protecting sensitive user interactions. Our work makes four primary contributions. First, we identified Hubness Explosion---the unbounded local sensitivity caused by hub nodes in scale-free networks---as the critical barrier preventing effective differential privacy for graph-supervised hashing. Second, we proposed a Sanitize-then-Distill paradigm that decouples privacy preservation from representation learning through deterministic degree clipping, Noisy Mirror Descent synthesis, and rectified log-normalization. Third, we introduced a Holistic Structural Loss that enforces both intra-modal and cross-modal alignment, resolving the modality misalignment problem common in unsupervised hashing. Fourth, we provided theoretical privacy and utility guarantees, showing that the synthetic graph error scales sub-linearly with dataset size under bounded local sensitivity.
% Extensive experiments on MIRFlickr-25K and NUS-WIDE demonstrated that DMP-MH outperforms existing private baselines by 4.2--11.4 mAP points while retaining over 90\% of non-private state-of-the-art performance at a moderate privacy budget ($\epsilon=2.0$). 
Notably, DMP-MH exhibits graceful degradation in strict privacy regimes ($\epsilon=0.1$), preserving 77\% utility where competing methods collapse. Our framework additionally achieves $2.1\times$ faster training and 48\% lower memory consumption than DP-SGD, enabling practical deployment. 
% Future work will explore adaptive clipping strategies, extension to other modalities (audio, video), and integration with federated learning for distributed privacy-preserving retrieval. We believe DMP-MH establishes a principled foundation for deploying cross-modal retrieval systems that respect user privacy without sacrificing retrieval quality.
\clearpage
\bibliographystyle{ACM-Reference-Format}
\bibliography{sample-base}
\end{document}